\documentclass[submission,copyright,creativecommons]{eptcs}
\providecommand{\event}{TARK 2023} %
\usepackage{underscore}

\newenvironment{proof}{\paragraph{Proof.}}{\hfill$\square$}

\usepackage{hyperref}

\usepackage{url}

\usepackage{tikz}
\usetikzlibrary{positioning}

\usepackage[ruled,vlined,linesnumbered,lined]{algorithm2e}

\newcommand{\play}{\mathit{play}}

\usepackage{amsmath}
\usepackage{amsfonts}
\usepackage{amssymb}
\usepackage{latexsym}
\usepackage{alltt}
\usepackage{thm-restate}
\usepackage{graphics}
\usepackage{subcaption}

\usepackage{xspace}
\usepackage{graphicx,color}

\usepackage{enumerate}
\usepackage{color}
\usepackage{sgame}
\usepackage{microtype}

\usepackage{thmtools,thm-restate}

\newcommand{\turn}{\mathit{turn}}

\newcommand{\leaf}{\mathit{leaf}}

\newcommand{\spe}{\mathit{SPE}}

\newcommand{\ibar}{-i}

\newcommand{\wmax}{\mathit{win}}
\newcommand{\lose}{\mathit{lose}}
\newcommand{\pmax}{\mathit{max}}

\makeatletter
\newcommand*{\rom}[1]{\expandafter\@slowromancap\romannumeral #1@}
\makeatother

\title{Iterated Elimination of Weakly Dominated Strategies in Well-Founded Games}
\author{Krzysztof R. Apt
	\institute{Centrum Wiskunde \& Informatica\\ Amsterdam, The Netherlands}
	\institute{University of Warsaw\\ Warsaw, Poland}
    \email{k.r.apt@cwi.nl}
	\and
        Sunil Simon
	\institute{Department of CSE,\\ IIT Kanpur, Kanpur, India} 
	\email{simon@cse.iitk.ac.in}
}
\def\titlerunning{Iterated elimination of weakly dominated strategies in well-founded games}
\def\authorrunning{K.~R.~Apt \& S.~Simon}

\begin{document}

\newcommand{\Proof}{\NI
                    {\bf Proof.}\ }

\newtheorem{theorem}{Theorem}
\newtheorem{defined}[theorem]{Definition}
\newenvironment{definition}{\begin{defined} \rm}{\end{defined}}
\newtheorem{exa}[theorem]{Example}
\newenvironment{example}{\begin{exa} \rm}{\end{exa}}
\newtheorem{proposition}[theorem]{Proposition}
\newtheorem{conjecture}[theorem]{Conjecture}
\newtheorem{lemma}[theorem]{Lemma}
\newtheorem{corollary}[theorem]{Corollary}
\newtheorem{remark}[theorem]{Remark}
\newtheorem{note}[theorem]{Note}
\newtheorem{fact}[theorem]{Fact}
\newtheorem{claim}[theorem]{Claim}
\newtheorem{exe}{Exercise}
\newenvironment{exercise}{\begin{exe} \rm }{\end{exe}}
\newtheorem{pro}{Problem}
\newenvironment{problem}{\begin{pro} \rm }{\end{pro}}

\newcommand{\myparpic}[1]{\parpic{{\darkgray{#1}}}}
\newcommand{\bfei}[1]{\begin{bfseries}\emph{#1}\end{bfseries}\index{#1}}
\newcommand{\bfe}[1]{\begin{bfseries}\emph{#1}\end{bfseries}}
\newcommand{\inv}{\invisible}
\newcommand{\T}{\mbox{{\bf true}}}
\newcommand{\F}{\mbox{{\bf false}}}
\newcommand{\ES}{\mbox{$\emptyset$}}

\newcommand{\myla}{\mbox{$\:\leftarrow\:$}}
\newcommand{\myra}{\mbox{$\:\rightarrow\:$}}
\newcommand{\da}{\mbox{$\:\downarrow\:$}}
\newcommand{\ua}{\mbox{$\:\uparrow\:$}}
\newcommand{\La}{\mbox{$\:\Leftarrow\:$}}
\newcommand{\Ra}{\mbox{$\:\Rightarrow\:$}}
\newcommand{\tra}{\mbox{$\:\rightarrow^*\:$}}
\newcommand{\lra}{\mbox{$\:\leftrightarrow\:$}}
\newcommand{\bu}{\mbox{$\:\bullet \ \:$}}
\newcommand{\up}{\mbox{$\!\!\uparrow$}}

\newcommand{\A}{\mbox{$\ \wedge\ $}}
\newcommand{\bigA}{\mbox{$\bigwedge$}}
\newcommand{\Orr}{\mbox{$\ \vee\ $}}
\newcommand{\U}{\mbox{$\:\cup\:$}}
\newcommand{\I}{\mbox{$\:\cap\:$}}
\newcommand{\sse}{\mbox{$\:\subseteq\:$}}

\newcommand{\po}{\mbox{$\ \sqsubseteq\ $}}
\newcommand{\spo}{\mbox{$\ \sqsubset\ $}}
\newcommand{\rpo}{\mbox{$\ \sqsupseteq\ $}}
\newcommand{\rspo}{\mbox{$\ \sqsupset\ $}}

\newcommand{\Mo}{\mbox{$\:\models\ $}}
\newcommand{\mtwo}{\mbox{$\:\models_{\it 2}\ $}}
\newcommand{\mthree}{\mbox{$\:\models_{\it 3}\ $}}
\newcommand{\Mf}{\mbox{$\:\models_{\it f\!air}\ $}}
\newcommand{\Mt}{\mbox{$\:\models_{\it tot}\ $}}
\newcommand{\Mwt}{\mbox{$\:\models_{\it wtot}\ $}}
\newcommand{\Mp}{\mbox{$\:\models_{\it part}\ $}}

\newcommand{\PR}{\mbox{$\vdash$}}

\newcommand{\fa}{\mbox{$\forall$}}
\newcommand{\te}{\mbox{$\exists$}}
\newcommand{\fai}{\mbox{$\stackrel{\infty}{\forall}$}}
\newcommand{\tei}{\mbox{$\stackrel{\infty}{\exists}$}}

\newcommand{\calA}{\mbox{$\cal A$}}
\newcommand{\calG}{\mbox{$\cal G$}}
\newcommand{\calB}{\mbox{$\cal B$}}

\newcommand{\LLn}{\mbox{$1,\ldots,n$}}
\newcommand{\LL}{\mbox{$\ldots$}}

\newcommand{\C}[1]{\mbox{$\{{#1}\}$}}           
\newcommand{\NI}{\noindent}
\newcommand{\HB}{\hfill{$\Box$}}
\newcommand{\VV}{\vspace{5 mm}}
\newcommand{\III}{\vspace{3 mm}}
\newcommand{\II}{\vspace{2 mm}}
\newcommand{\PP}{\mbox{$[S_1 \| \LL \| S_n]$}}

\newcommand{\X}[1]{\mbox{$\:\stackrel{{#1}}{\rightarrow}\:$}}
\newcommand{\lX}[1]{\mbox{$\:\stackrel{{#1}}{\longrightarrow}\:$}}

\newenvironment{mydef}{\begin{boxitpara}{box 0.95 setgray fill}}{\end{boxitpara}}

\newcommand{\vect}[1]{{\bf #1}}
\newcommand{\D}[1]{\mbox{$|[{#1}]|$}}           
\newcommand{\hb}[1]{{\Theta_{#1}}}
\newcommand{\hi}{{\cal I}}
\newcommand{\ran}{\rangle}
\newcommand{\lan}{\langle}
\newcommand{\dom}{{\it Dom}}
\newcommand{\var}{{\it Var}}
\newcommand{\pred}{{\it pred}}
\newcommand{\sem}[1]{{\mbox{$[\![{#1}]\!]$}}}
\newcommand{\false}{{\it false}}
\newcommand{\true}{{\it true}}
\newtheorem{Property}{Property}[section]
\newcommand{\hsuno}{\hspace{ .5in}}
\newcommand{\vsuno}{\vspace{ .25in}}
\newcommand{\cons}{$\! \mid $}
\newcommand{\nil}{[\,]}
\newcommand{\restr}[1]{\! \mid \! {#1}} 
\newcommand{\range}{{\it Ran}}
\newcommand{\Range}{{\it Range}}
\newcommand{\PC}{\mbox{$\: \simeq \:$}}
\newcommand{\TC}{\mbox{$\: \simeq_t \:$}}
\newcommand{\IO}{\mbox{$\: \simeq_{i/o} \:$}}
\newcommand{\size}{{\rm size}}
\newcommand{\Der}[2]
        {\; |\stackrel{#1}{\!\!\!\longrightarrow _{#2}}\; }

\newcommand{\Sder}[2]
        {\; |\stackrel{#1}{\!\leadsto_{#2}}\; }

\newcommand{\szkew}[1]{\relax \setbox0=\hbox{\kern -24pt $\displaystyle#1$\kern 0pt }%
\box0}
{\catcode`\@=11 \global\let\ifjusthvtest@=\iffalse}

\newcommand{\iif}{\mbox{$\longrightarrow$}}
\newcommand{\longra}{\mbox{$\longrightarrow$}}
\newcommand{\Longra}{\mbox{$\Longrightarrow$}}
\newcommand{\res}[3]{\mbox{$#1\stackrel{\textstyle #2}{\Longra}#3$}}
\newcommand{\reso}[4]{\mbox{
$#1 \stackrel{#2}{\Longra}_{\hspace{-4mm}_{#3}} \hspace{2mm} #4$}}
\newcommand{\nextres}[2]{\mbox{$\stackrel{\textstyle #1}{\Longra}#2$}}
\newcommand{\preres}[2]{\mbox{$#1\stackrel{\textstyle #2}{\Longra}$}}
\newcommand{\Preres}[2]{\mbox{$#1\stackrel{\textstyle #2}{\longra}$}}

\newcounter{oldmycaption}
\newcommand{\oldmycaption}[1]
{\begin{center}\addtocounter{oldmycaption}{1}
{{\bf Program \theoldmycaption} : #1}\end{center}} 

\newcommand{\mycite}[1]{\cite{#1}\glossary{#1}}

\newcommand{\ass}{{\cal A}}
\newcommand{\defi}{{\stackrel{\rm def}{=}}}

\newcommand{\p}[2]{\langle #1 \ ; \ #2 \rangle}

\newcommand{\ceiling}[1]{\lceil #1 \rceil}
\newcommand{\adjceiling}[1]{\left \lceil #1 \right \rceil}
\newcommand{\floor}[1]{\lfloor #1 \rfloor}
\newcommand{\adjfloor}[1]{\left \lfloor #1 \right \rfloor}

\maketitle

\begin{abstract}
  Recently, in \cite{AS21}, we studied well-founded games, a natural
  extension of finite extensive games with perfect information in which all plays are finite.
  We extend here, to this class of games, two results concerned with
  iterated elimination of weakly dominated strategies, originally
  established for finite extensive games.

  The first one states that every finite extensive game with perfect
  information and injective payoff functions can be reduced by a specific
  iterated elimination of weakly dominated strategies to a trivial
  game containing the unique subgame perfect equilibrium. Our
  extension of this result to well-founded games admits transfinite
  iterated elimination of strategies.  It applies to an infinite
  version of the centipede game.  It also generalizes the original
  result to a class of finite games that may have several subgame perfect equilibria.

  The second one states that finite zero-sum games with $n$ outcomes
  can be solved by the maximal iterated elimination of weakly
  dominated strategies in $n-1$ steps.  We generalize this result to
  a natural class of well-founded strictly competitive games.
\end{abstract}

 \section{Introduction}
\label{sec:intro}
This paper is concerned with the iterated elimination of weakly
dominated strategies (IEWDS) in the context of natural class of
infinite extensive games with perfect information.
While simple examples show that the deletion of weakly dominated
strategies may result in removal of a unique Nash equilibrium,
IEWDS has some merit if it results in solving a game.  It is for instance used to show 
that the so-called ``beauty contest'' game has exactly one Nash equilibrium
(see, e.g., \cite[Chapter 5]{Hei12}).  Other games can be solved this
way, see, e.g., \cite[pages 63, 110-114]{OR94}.

This procedure was also studied in the realm of finite extensive games
with perfect information.  In \cite{Hum08} the correspondence between
the outcomes given by the iterated elimination of weakly dominated
strategies and backward induction was investigated in the context of
binary voting agendas with sequential voting. More recently, this procedure
was studied in \cite{Sob19} in the context of supermodular games.

For arbitrary games two
important results were established.  The first one states, see
\cite{OR94}, that in such games with injective payoff functions (such
games are sometimes called \emph{generic}) a specific iterated
elimination of weakly dominated strategies (that mimics the backward
induction) yields a trivial game which contains the unique subgame
perfect equilibrium.  It was noticed in \cite{Bat97} that this result
holds for a slightly more general class of games \emph{without
  relevant ties}.\footnote{All mentioned concepts are explained in
  Sections \ref{sec:preliminaries}, \ref{sec:main}, and \ref{sec:strictly}.  
  We did not find any precise proofs in the literature.  The proof is briefly
  sketched in \cite[pages 108-109]{OR94} and summarized in \cite[pages
  48-49]{Bat97} as follows: ``if backward induction deletes action $a$
  at node $x$, delete all the strategies reaching $x$ and choosing
  $a$''.  We provided in \cite{AS21a} a detailed proof of the stronger
  result of \cite{Bat97} in which we clarified how the backward
  induction algorithm needs to be modified to achieve the desired
  outcome.}

The second result, due to \cite{Ewe02}, is concerned with finite
extensive zero-sum games.  It states that such games can be reduced to
a trivial game by the `maximal' iterated elimination of weakly
dominated strategies in $n-1$ steps, where $n$ is the number of
outcomes.\footnote{An alternative proof given in \cite{Ost05} shows
that the result holds for the larger class of strictly competitive games. In
\cite{AS21a} we clarified that the original proof also holds for this class of games.}

In \cite{AS21} we studied a natural extension of finite extensive
games with perfect information in which one assumes that all plays are
finite. We called these games well-founded games.\footnote{In the
  economic literature such games are sometimes called `games with
  finite horizon'.}  The subject of this paper is to extend the above
two results to well-founded games.
In both cases some non-trivial difficulties arise.

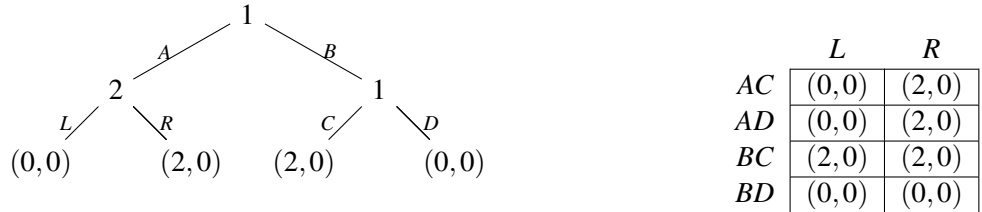
\begin{figure}[h]
\centering
\begin{minipage}{.75\textwidth}
  \centering  
  \tikzstyle{level 1}=[level distance=1cm, sibling distance=3.5cm]
  \tikzstyle{level 2}=[level distance=1cm, sibling distance=2cm]
  \tikzstyle{level 3}=[level distance=1.5cm, sibling distance=2cm]
\begin{tikzpicture}
 \node (r){1}
 child{
   node (a){2}
   child{
     node (d){$(0, 0)$}
     edge from parent
     node[left]{\scriptsize $L$}
   }
   child{
     node(e){$(2,0)$}
     edge from parent
     node[right]{\scriptsize $R$}
     edge from parent
   }
   edge from parent
   node[left]{\scriptsize $A$}
   }
 child{
   node (b){1}
   child{
     node (f){$(2,0)$}
     edge from parent
     node[left]{\scriptsize $C$}
   }
   child{
     node (g){$(0,0)$}
     edge from parent
     node[right]{\scriptsize $D$}
     edge from parent
   }
   edge from parent
   node[right]{\scriptsize $B$}
 };
\end{tikzpicture}

\captionof{figure}{An extensive  game $G$ and the corresponding strategic game $\Gamma(G)$}
  \label{fig:speEquiv}
\end{minipage}%
\begin{minipage}{.25\textwidth}
  \centering
\begin{tabular}{r|c|c|}
\multicolumn{1}{r}{}
 &  \multicolumn{1}{c}{$L$}
 &  \multicolumn{1}{c}{$R$}\\
\cline{2-3}
\textit{AC} & $(0,0)$ & $(2,0)$\\
\cline{2-3}
\textit{AD} & $(0,0)$ & $(2,0)$\\
\cline{2-3}
\textit{BC} & $(2,0)$ & $(2,0)$\\
\cline{2-3}
\textit{BD} & $(0,0)$ & $(0,0)$\\
\cline{2-3}
\end{tabular}
  
\end{minipage}
\end{figure}

\begin{example}
\label{ex:speEquiv}
Consider the extensive game $G$ and the corresponding strategic game
$\Gamma(G)$ given in Figures \ref{fig:speEquiv}. $G$ has three subgame perfect
equilibria which are all payoff equivalent: $\{(\mathit{AC},R),
(\mathit{BC},L), (\mathit{BC},R)\}$.
We can observe that in $\Gamma(G)$ no sequence of
iterated elimination of weakly dominated strategies results in a
trivial game that contains all the subgame perfect equilibria in
$G$. To see this, first note that the strategies $L$ and $R$ of player
2 are never weakly dominated irrespective of the elimination done with
respect to the strategies of player 1. Also, note that the strategy
$\mathit{BD}$ of player 1 is strictly dominated by $\mathit{BC}$ in
$\Gamma(G)$. Thus the only possibility of reducing $\Gamma(G)$ to a
trivial game is to eliminate all strategies of player 1 except
$\mathit{BC}$. But this results in the elimination of
$(\mathit{AC}, R)$ which is a subgame perfect equilibrium in $G$.
\HB
\end{example}

This might suggest that one should limit oneself to extensive games
with a unique subgame perfect equilibrium.  Unfortunately, this
restriction does not work either as shown in Example
\ref{ex:1spe}. Additional complication arises when the game has no
subgame perfect equilibrium as shown in \ref{ex:nospe}.

\begin{figure}[h]
\centering
\begin{minipage}{.45\textwidth}
  \centering  
\tikzstyle{level 1}=[level distance=1.5cm, sibling distance=2.5cm]
  \tikzstyle{level 2}=[level distance=1.5cm, sibling distance=1.5cm]
  \tikzstyle{level 3}=[level distance=1.5cm, sibling distance=2cm]
\begin{tikzpicture}
 \node (r){1}
 child{
   node (a){$(0,100)$}
   edge from parent
   node[left]{\scriptsize $0$}
   }
 child{
   node (b){$(x,100-x)$}
   edge from parent
   node[left]{\scriptsize $x$}
 }
  child{
   node (c){2}
   child{
     node (h){$(100, 0)$}
     edge from parent
     node[left]{\scriptsize $L$}
   }
   child{
     node (i){$(0,0)$}
     edge from parent
     node[right]{\scriptsize $R$}
     edge from parent
   }
  edge from parent
  node[right]{\scriptsize $100$}};

  \path (a) -- (b) node [midway] {$\cdots$};
  \path (b) -- (c) node [midway] {$\cdots$};
\end{tikzpicture}

\captionof{figure}{A game $G$ with a unique SPE}
\label{fig:ultimatum1}
\end{minipage}%
\begin{minipage}{.5\textwidth}
  \centering
\begin{tikzpicture}
    [
 scale=1.5,font=\footnotesize,
 level 1/.style={level distance=10mm,sibling distance=15mm},
 level 2/.style={level distance=10mm,sibling distance=8mm},
]
 \node (r){1}
 child{
   node (a){$(0,0)$}
   edge from parent
   node[left]{$A$}
 }
   child{
   node (a){$(0,1)$}
   edge from parent
   node[left]{$B$}
   }
 child{
   node (b){$2$}
   child{
     node (f){$(0,0)$}
     edge from parent
     node[left]{}
   }
   child{
     node (g){$(0,1)$}
     edge from parent
     node[right]{}
   }
   child{
     node (f){$(0,2)$}
     edge from parent
     node[left]{}
   }
   child {
     node {$\ldots$}  edge from parent[draw=none]
     node[right]{$\cdots$}}
       edge from parent
   node[right]{ $C$}
 };
\end{tikzpicture}
\captionof{figure}{A game $G$ with no SPE}
  \label{fig:nospe}
\end{minipage}
\end{figure}

\begin{example}
\label{ex:1spe}

Consider a `trimmed version' of the ultimatum game from \cite{AS21}
given in Figure \ref{fig:ultimatum1}, in which for each
$x \in [0,100]$ the root has a direct descendant $x$.
This game has a unique subgame perfect equilibrium, namely $(100, L)$.
Consider an iterated elimination of weakly dominated strategies.
For each strategy of player 1 the strategies $L$ and $R$ of player 2
yield the same payoff.  So these two strategies are never eliminated.
Further, strategy 100 of player 1 is never eliminated either, since
for any strategy $x < 100$ we have $p_1(x, L) = x < 100 = p_1(100, L)$
and $p_1(x, R) = x > 0 = p_1(100, R)$. So the joint strategies
$(100, L)$ and $(100, R)$ are never eliminated and they are not payoff
equivalent.
(In fact, each iterated elimination of weakly dominated strategies yields the game with the sets of strategies $\{100\}$ and  $\{L, R\}$.)
\HB
\end{example}

\begin{example}
\label{ex:nospe}
Consider the well-founded game $G$ given in Figure
\ref{fig:nospe}. Clearly $G$ has no subgame perfect
equilibrium. Further, strategies $A$ and $B$ of player 1 yield the
same outcome for him, so cannot be eliminated by any iterated
elimination of weakly dominated strategies. Thus any result
of such an elimination contains at least two outcomes, $(0,0)$ and
$(0,1)$.  So $G$ cannot be reduced to a trivial game.
\HB
\end{example}

To address these issues, we introduce the concept of an
\emph{SPE-invariant} well-founded game. These are games in which
subgame perfect equilibria exist and moreover in each subgame such
equilibria are payoff equivalent. Then we show that the first result
can be extended to such games.  In view of the above examples it
looks like the strongest possible generalization of the original
result. In particular, it applies to an infinite version of the
well-known centipede game of \cite{Ros81}.

This result calls for a careful extension of the iterated elimination
of weakly dominated strategies to infinite games: its stages have to
be indexed by ordinals and one has to take into account that the
outcome can be the empty game.

When limited to finite games, our theorem extends
the original result. In particular it applies to the class of
extensive games that satisfy the \bfe{transference of decisionmaker
  indifference (TDI)} condition due to \cite{MS97}, a class
that includes strictly competitive games.  We also show that the
well-founded games with finitely many outcomes that satisfy the TDI
condition are SPE-invariant. 
Also when extending the second result, about strictly competitive
games, to well-founded games one has to be careful. The original proof
crucially relies on the fact that finite extensive zero-sum games have
a value.  Fortunately, as we showed in \cite{AS21}, well-founded games
with finitely many outcomes have a subgame perfect equilibrium, so a
fortiori a Nash equilibrium, which suffices to justify the relevant
argument (Lemma \ref{lem:h1} in Section \ref{sec:strictly}).

By carefully checking of the crucial steps of the original proof we
extend the original result to a class of well-founded strictly
competitive games that includes \emph{almost constant} games, in
which for all but finitely many leaves the outcome is the same.  It
remains an open problem whether this result holds for all strictly
competitive games with finitely many outcomes.

IEWDS is one of the early approaches applied to analyze strategies and
extensive games. It does not take into account epistemic reasoning of
players in the presence of assumptions such as common knowledge of
rationality. The vast literature on this subject, starting with
\cite{Ber84} and \cite{Pea84}, led to identification of several more
informative ways of analyzing finite extensive games with imperfect
information. We just mention here two representative references.  In
\cite{Bat97} Pearce's notion of \emph{extensive form
rationalizability} (EFR) was studied and it was shown that for
extensive games without relevant ties it coincides with the IEWDS. A
more general notion of common belief in future rationality was studied
in \cite{Per14} that led to identification of a new iterative
elimination procedure called \emph{backward dominance}.

In our paper IEWDS is defined as a transfinite elimination
procedure. A number of papers, starting with \cite{Lip94}, analyzed
when such a transfinite elimination of strategies cannot be reduced to
an iteration over $\omega$ steps. In our framework it is a simple
consequence of the fact that the ranks of the admitted game trees can
be arbitrary ordinals. In particular, an infinite version of the centipede
game considered in Example \ref{exa:centipede}  requires more than $\omega$
elimination rounds.

 \section{Preliminaries}
\label{sec:preliminaries}

\subsection{Strategic games}

A \bfe{strategic game} $H=(H_1, \LL, H_n, p_1, \LL, p_n)$ consists of
a set of players $\{1, \LL, n\}$, where $n \geq 1$, and for each
player $i$, a set $H_i$ of \bfe{strategies} along with a \bfe{payoff
  function} $p_i : H_1 \times \cdots \times H_n \myra \mathbb{R}$.







We call each element of $H_1 \times \cdots \times H_n$ a \bfe{joint
  strategy} of players $1, \LL, n$, denote the $i$th element of
$s \in H_1 \times \cdots \times H_n$ by $s_i$, and abbreviate the
sequence $(s_{j})_{j \neq i}$ to $s_{-i}$. We write $(s'_i, s_{-i})$
to denote the joint strategy in which player's $i$ strategy is $s'_i$
and each other player's $j$ strategy is $s_j$.  Occasionally we write
$(s_i, s_{-i})$ instead of $s$.  Finally, we abbreviate the Cartesian
product $\times_{j \neq i} H_j$ to $H_{-i}$.

Given a joint strategy $s$, we denote the sequence
$(p_1(s), \LL, p_n(s))$ by $p(s)$ and call it an \bfe{outcome} of the
game.  We say that $H$ \bfe{has $k$ outcomes} if
$|\{p(s) \mid s \in H_1 \times \cdots \times H_n \}|= k$ and call a
game \bfe{trivial} if it has one outcome. If one of the sets $H_i$ is
empty, we call the game \bfe{empty} and \bfe{non-empty}
otherwise. Unless explicitly stated, all used strategic games are
assumed to be non-empty.  We say that two joint strategies $s$ and $t$
are \bfe{payoff equivalent} if $p(s) = p(t)$.

We call a joint strategy $s$ a \bfe{Nash equilibrium} if 
$
\fa i \in \{1, \ldots, n\}  \fa s'_i \in H_i : p_i(s_i, s_{-i}) \geq p_i(s'_i, s_{-i})
$.
When the number of players and their payoff functions are known we can
identify the game $H$
with the set of strategies in it.

By a \bfe{subgame} of a strategic game $H$ we mean a game obtained
from $H$ by removing some strategies.  Given a set $\cal J$ of
subgames of a strategic game $H$ we define $\bigcap {\cal J}$ as the
subgame of $H$ in which for each player $i$ his set of strategies is
$ \bigcap_{J \in {\cal J}} J_i$.
Also, given two subgames $H'$ and $H''$ of a strategic game $H$ we
write $H' \sse H''$ if for each player $i$, $H'_i \sse H''_i$.

Consider two strategies $s_i$ and $s'_i$ of player $i$ in a strategic
game $H$.  We say that $s_i$ \bfe{weakly dominates $s'_i$} (or
equivalently, that $s'_i$ is \bfe{weakly dominated by $s_i$}) in $H$
if $\fa s_{-i} \in H_{-i} : p_i(s_i, s_{-i}) \geq p_i(s'_i, s_{-i})$
and $\te s_{-i} \in H_{-i} : p_i(s_i, s_{-i}) > p_i(s'_i, s_{-i})$.

In what follows, given a strategic game we consider, possibly
transfinite, sequences of sets of strategies.  They are written as
$(\rho_{\alpha}, \alpha < \gamma)$, where $\alpha$ ranges over all
ordinals smaller than some ordinal $\gamma$.
%
%
Given two such sequences $\rho := (\rho_{\alpha}, \alpha < \gamma)$
and $\rho' := (\rho'_{\alpha'}, \alpha' < \gamma')$, we denote by
$(\rho, \rho')$ their concatenation (which is indexed by
$\gamma + \gamma'$), by $\rho^{\beta}$ the subsequence
$(\rho_{\alpha}, \alpha < \beta)$ of $\rho$,
and for $\alpha < \beta$ by $\rho^{\beta - \alpha}$ the subsequence
such that $(\rho^{\alpha}, \rho^{\beta - \alpha}) = \rho^{\beta}$.
Further, we write $H \to^{\rho} H'$ to denote the fact that the game
$H'$ is the outcome of the iterated elimination from the non-empty
game $H$ of the sets of strategies that form $\rho$.  In each step all
eliminated strategies are weakly dominated in the current game. As a
result $H'$ may be empty.  The relation $\to^{\rho}$ is defined as
follows.

If $\rho = (\rho_0)$, that is, if $\gamma = 1$, then $H \to^{\rho} H'$
holds if each strategy in the set $\rho_0$ is weakly dominated in $H$
and $H'$ is the outcome of removing from $H$ all strategies from
$\rho_0$.
If $\gamma$ is a successor ordinal $>1$, say
$\gamma = \delta + 1$, and $H \to^{\rho'} H'$,
$H' \to^{(\rho_{\delta})} H''$, where $H'$ is non-empty, and
$\rho' := (\rho_{\alpha}, \alpha < \delta)$,
then $H \to^{\rho} H''$.
%
Finally, if $\gamma$ is a limit ordinal and for all
$\beta < \gamma$, $H \to^{\rho^{\beta}} H^{\beta}$,
then $H \to^{\rho} \bigcap_{\beta < \gamma} H^{\beta}$.
In general, the strategic game $H$ from which we eliminate strategies
will be a subgame of a game $\Gamma(G)$, where $G$ is an extensive
game (to be defined shortly).  It will be then convenient to allow in
$\rho$ strategies from $\Gamma(G)$. In the definition of
$H \to^{\rho} H'$ we then disregard the strategies from $\rho$ that
are not from $H$.
In the proofs below we rely on the following observations about the
$\to^{\rho}$ relation, the proofs of which we omit.

\begin{note} \label{not:one}
  \mbox{} \hspace{-4mm}
  \begin{enumerate}[(i)]
  \item Suppose $H \to^{\rho} H'$ and $H' \to^{\rho'} H''$, where $H'$ is non-empty. Then
    $H \to^{(\rho, \rho')} H''$.

    \item Suppose $H \to^{\rho} H'$, where
    $\rho = (\rho_{\alpha}, \alpha < \gamma)$ and $\gamma$ is a limit
    ordinal. Suppose further that for a sequence of ordinals
    $(\alpha_{\delta})_{\delta < \epsilon}$ converging to $\gamma$ we
    have $H \to^{\rho^{\alpha_{\delta}}} H^{\alpha_{\delta}}$
     for all $\delta < \epsilon$.  Then
    $H' = \bigcap_{\delta < \epsilon} H^{\alpha_{\delta}}$. 
  \end{enumerate}
\end{note}

\subsection{Well-founded games}

We recall from \cite{AS21} the definition of a well-founded game.
A \bfe{tree} is an acyclic directed connected graph, written as
$(V,E)$, where $V$ is a non-empty set of nodes and $E$ is a possibly
empty set of edges.
An \bfe{extensive game with perfect information} $(T, \turn, p_1, \LL,
p_n)$ consists of a set of players $\{1, \LL, n\}$, where $n \geq 1$
along with the following. A \bfe{game tree}, which is a tree $T :=
(V,E)$ with a \bfe{turn function} $\turn: V \setminus Z \to \{1, \LL,
n\}$, where $Z$ is the set of leaves of $T$. For each player $i$ a
\bfe{payoff function} $p_i: Z \myra \mathbb{R}$, for each player $i$.
%
%
%
The function $\turn$ determines at each non-leaf node which player
should move. The edges of $T$ represent possible \bfe{moves} in the
considered game, while for a node $v \in V \setminus Z$ the set of its
children $C(v) := \{w \mid (v,w) \in E\}$ represents possible
\bfe{actions} of player $\turn(v)$ at $v$. 

We say that an extensive game with perfect information is
\bfe{finite}, \bfe{infinite}, or \bfe{well-founded} if, respectively,
its game tree is finite, infinite, or
well-founded. Recall that a tree is called \bfe{well-founded} if it
has no infinite paths.  \emph{From now on by an \bfe{extensive game} we mean
a well-founded extensive game with perfect information.}

For a node $u$ in $T$ we denote the subtree of $T$ rooted at $u$ by
$T^u$.  In the proofs we shall often rely on the concept of a
\emph{rank} of a well-founded tree $T$, defined inductively as
follows, where $v$ is the root of $T$:
  \[
\mathit{rank}(T):=
  \begin{cases}
    0 &\text{ if $T$ has one node}\\
    \mathit{sup} \{\mathit{rank}(T^u) + 1 \mid u \in C(v) \} &\text{
      otherwise,}
  \end{cases}
\]
where $sup(X)$ denotes the least ordinal larger than all ordinals in
the set $X$.

For an extensive game $G:= (T, \turn, p_1, \LL, p_n)$ let 
$V_i := \{v \in V \setminus Z \mid \turn(v) = i\}$. So $V_i$ is the
set of nodes at which player $i$ moves.  A \bfe{strategy} for player
$i$ is a function $s_i: V_i \to V$, such that $(v, s_i(v)) \in E$ for
all $v \in V_i$.  We denote the set of strategies of player $i$ by
$S_i$. 
Let $S = S_1 \times \cdots \times S_n$. As in the case of the
strategic games we use the `$-i$' notation, when referring to
sequences of strategies or sets of strategies.

Each joint strategy $s = (s_1, \LL, s_n)$ determines a rooted path
$\mathit{play}(s) := (v_1, \LL, v_m)$ in $T$ defined inductively as
follows. $v_1$ is the root of $T$ and if $v_{k} \not\in Z$, then $v_{k+1} := s_i(v_k)$, where $\turn(v_k) = i$.
So when the game tree consists of just one node, $v$, we have
$\mathit{play}(s) = v$.  Informally, given a joint strategy $s$, we
can view $\mathit{play}(s)$ as the resulting \emph{play} of the game.
For each joint strategy $s$ the rooted path $\mathit{play}(s)$ is
finite since the game tree is assumed to be well-founded.  Denote by
$\leaf(s)$ the last element of $\mathit{play}(s)$.  To simplify the
notation we just write everywhere $p_i(s)$ instead of $p_i(\leaf(s))$.

With each extensive game $G: = (T, \turn, p_1, \LL, p_n)$ we associate
a strategic game $\Gamma(G)$ defined as follows. $\Gamma(G):= (S_1,
\LL, S_n , p_1, \LL, p_n)$, where each $S_i$ is the set of strategies
of player $i$ in $G$.  In the degenerate situation when the game tree
consists of just one node, each strategy is the empty function,
denoted by $\ES$, and there is only one joint strategy, namely the
$n$-tuple $(\ES, \LL, \ES)$ of these functions. In that case we just
stipulate that $p_i(\ES, \LL, \ES) = 0$ for all players $i$.
All notions introduced in the context of strategic games can now be
reused in the context of an extensive game $G$ simply by referring to
the corresponding strategic form $\Gamma(G)$.  In particular, the
notion of a Nash equilibrium is well-defined.

The \bfe{subgame} of an extensive game $G:= (T, \turn, p_1, \LL,
p_n)$, rooted at the node $w$ and denoted by $G^w$, is defined as
follows. The set of players is $\{1, \LL, n\}$, the game tree is $T^w$. The
$\turn$ and payoff functions are the restrictions of the corresponding
functions of $G$ to the nodes of $T^w$.
We call $G^{w}$ a \bfe{direct subgame} of $G$ if $w$ is a child of
the root $v$.

Note that some players may `drop out' in $G^w$, in the sense that at
no node of $T^w$ it is their turn to move.  Still, to keep the
notation simple, it is convenient to admit in $G^w$ all original
players in $G$.

Each strategy $s_i$ of player $i$ in $G$ uniquely determines his
strategy $s^w_i$ in $G^w$.  Given a joint strategy
$s = (s_1, \LL, s_n)$ of $G$ we denote by $s^w$ the joint strategy
$(s^w_1, \LL, s^w_n)$ in $G^w$.  Further, we denote by $S_i^w$ the set
of strategies of player $i$ in the subgame $G^w$ and by $S^w$ the set
of joint strategies in this subgame.

Finally, a joint strategy $s$ of $G$ is called a \bfe{subgame perfect
  equilibrium} in $G$ if for each node $w$ of $T$, the joint strategy
$s^w$ of $G^w$ is a Nash equilibrium in the subgame $G^w$.  We denote
by $\spe(G)$ the set of subgame perfect equilibria in $G$.  Finally,
we say that a game is \bfe{SPE-invariant} if it has a subgame perfect
equilibrium and in each subgame of it all subgame perfect equilibria
are payoff equivalent.

We shall often use the following result.

\begin{theorem}[\cite{AS21}]
  \label{thm:constantSPE}
  Every extensive game with finitely many outcomes has a subgame
    perfect equilibrium.
\end{theorem}

 \section{Preliminary lemmas}
\label{sec:first}

In this section we present a sequence of lemmas needed to prove our
first main result. In the proofs we often switch between a game and
its direct subgames.

Consider an extensive game $G: = (T, \turn,
p_1, \LL, p_n)$ with the root $v$ and a child $w$ of $v$.
For each player $j$ to each of his strategy $t_j$ in a direct subgame
$G^w$ there corresponds a natural set $[t_j]$ of his strategies in the
game $G$ defined by
$
  [t_j] := \{s_j \mid t_j = s^w_j \text{\ and\ } s_j(v) = w \text{\ if\ } j = \turn(v)\}$.
So for a player $j$, $[t_j]$ is the set of his strategies in $G$
the restriction of which to $G^w$ is $t_j$, with the
additional proviso that if $j = \turn(v)$, then each strategy in
$[t_j]$ selects $w$ at the root $v$. We call $[t_j]$ the
\emph{lifting} of $t_j$ to the game $G$.
The following lemma clarifies the relevance of lifting.

\begin{restatable}{lemma}{lemLifting}
    \label{lem:lifting}
  Consider a direct subgame $G^w$ of $G$. Suppose that the
  strategy $t_j$ is weakly dominated in $G^w$. Then
  each strategy in $[t_j]$ is weakly dominated in $G$.
\end{restatable}

\begin{proof}
  Suppose that $t_j$ is weakly dominated in $G^w$ by some strategy
  $u_j$. Take a strategy $v_j$ in $[t_j]$.  We show that $v_j$ is
  weakly dominated in $G$ by the strategy $w_j$ in $[u_j]$ that
  coincides with $v_j$ on all the nodes that do not belong to
  $G^w$. So $w_j$ is obtained from $v_j$ by replacing in it $v^w_j$,
  i.e., $t_j$, by $u_j$. Below $s_{-j}$ denotes a sequence of
  strategies in $G$ of the opponents of player $j$.

  \medskip
  
  \NI  \emph{Case 1.} $j = \turn(v)$.
  
By the choice of $u_j$ for all $s_{-j}$ $ p_j(t_j, s^w_{-j}) \leq
p_j(u_j, s^w_{-j})$ and for some $s_{-j}$ $p_j(t_j, s^w_{-j}) <
p_j(u_j, s^w_{-j})$.  Further, by the definition of $[\cdot]$ we have
$v_j(v) = w$, so for all $s_{-j}$ we have $p_j(v_j, s_{-j})= p_j(t_j,
s^w_{-j})$ and $p_j(u_j, s^w_{-j}) = p_j(w_j, s_{-j})$, so the claim
follows.

\medskip

\NI  \emph{Case 2.} $j \neq \turn(v)$.

    Let $i = \turn(v)$. Take some $s_{-j}$.  If $s_i(v) = w$, then
    $p_j(v_j, s_{-j})= p_j(t_j, s^w_{-j})$ and $p_j(w_j, s_{-j})=
    p_j(u_j, s^w_{-j})$. Thus $p_j(v_j, s_{-j}) \leq p_j(w_j, s_{-j})$
    by the choice of $u_j$ and $w_j$.  Further, if $s_i(v) \neq w$,
    then $p_j(v_j, s_{-j})= p_j(w_j, s_{-j})$ by the choice of $w_j$.

  Choose an arbitrary $s_{-j}$ such that $s_i(v) = w$ and
  $p_j(t_j, s^w_{-j}) < p_j(u_j, s^w_{-j})$.  By the choice of $s_i$
  we have $p_j(v_j, s_{-j})= p_j(t_j, s^w_{-j})$ and $p_j(w_j,
  s_{-j})= p_j(u_j, s^w_{-j})$, so $p_j(v_j, s_{-j}) < p_j(w_j,
  s_{-j})$. Thus the claim follows.
\end{proof}

We now extend the notation $[\cdot]$ to sets of strategies and
sequences of sets strategies. First, given a set of strategies $A$ in
a direct subgame $G^w$ of $G$ we define
$[A] :=  \bigcup_{s_j \in A} [s_j]$.
Next, given a sequence $\rho$ of sets of strategies of players, each
set taken from a direct subgame of $G$, we denote by $[\rho]$ the
corresponding sequence of sets of strategies of players in $G$
obtained by replacing each element $A$ in $\rho$ by $[A]$.

Given a set $A$ of strategies of players
in a direct subgame $G^w$ we define the corresponding set of 
strategies in the game $G$ by putting
$\langle A \rangle = \{s_j \mid s^w_j \in A\}$.
Thus for a set $A$ of strategies in a direct subgame $G^w$, the set
$\langle A \rangle$ differs from $[A]$ in that we do include in the
former set strategies $s_j$ for which $s_j(v) \neq w$.  Given a set
$A$ of strategies of player $j$ in the subgame $G^w$, we call
$\langle A \rangle$ an \emph{extension} of $A$ to the game $G$.
Further, given a subgame $H$ of $\Gamma(G^w)$, we define
$\langle H \rangle$ as the subgame of $\Gamma(G)$ in which for each
player $j$ we have $\langle H \rangle_j = \langle H_j \rangle$.

In what follows we need a substantially strengthened version of
Lemma \ref{lem:lifting} that relies on the following concept.
Given an extensive game $G$ with a root $v$, we say that a
non-empty subgame $J$ of $\Gamma(G)$ \bfe{does not depend on} a direct subgame
$G^w$ if for any strategy $s_j$ from $J$ any modification of it
  on the non-leaf nodes of $G^w$ or on $v$ if $\turn(v) = j$ is also
  in $J$.
  Note that in particular $\Gamma(G)$ does not depend on any of its
  direct subgame and that for any non-empty subgame $H$ of a
  direct subgame $G^w$ of $G$ the subgame $\langle H \rangle$ does not
  depend on any other direct subgame of $G$.

  \begin{restatable}{lemma}{lemLiftingThree}
    \label{lem:lifting3}
  Consider a direct subgame $G^w$ of $G$, subgames $H$ and $H'$ of
  $\Gamma(G^w)$ and a set $A$ of strategies in $H$.  Suppose that $H
  \to^{A} H'$ and that the subgame $J$ of $\Gamma(G)$
  does not depend on $G^w$.  Then $J \cap \langle H \rangle \to^{[A]}
  J \cap \langle H' \rangle$.
  \end{restatable}

\begin{proof}
Take a strategy $v_j$ in $[A]$. For some strategy $t_j$ from $A$ that
is weakly dominated in $H$ by some strategy $u_j$ we have $v_j \in
[t_j] \cap J_j$.  Select a strategy $w_j$ in $[u_j]$ that coincides
with $v_j$ on the nodes that do not belong to $G^w$.  So $w_j$ is a
modification of $v_j$ on the non-leaf nodes of $G^w$ and consequently,
by the assumption about $J$, it is in $J_j$. Further, $w_j$ is in
$\langle H \rangle$, since $u_j$ is from $H$.

  We claim that $v_j$ is weakly dominated in $J \cap \langle H \rangle$ by $w_j$.  Below $s_{-j}$ denotes a sequence of strategies of the opponents of player $j$ in the original game $G$. 

  \medskip
  
  \NI\emph{Case 1.} $j = \turn(v)$.

By the choice of $u_j$ for all $s_{-j}$ such that $s^w_{-j} \in H_{-j}$
$p_j(t_j, s^w_{-j}) \leq p_j(u_j, s^w_{-j})$
and for some $s_{-j}$ such that $s^w_{-j} \in H_{-j}$
$p_j(t_j, s^w_{-j}) < p_j(u_j, s^w_{-j})$.
By the definition of `does not depend on' and the fact that $j = \turn(v)$
we can also assume that the latter 
$s_{-j}$ is from $J_{-j}$ by stipulating that $s_{-j} = t_{-j}$ for an arbitrary joint strategy $t$ from $J$.

Further, by the definition of $[\cdot]$ we have $v_j(v) = w$, so
for all $s_{-j}$ such that
$s^w_{-j} \in H_{-j}$ we have $p_j(v_j, s_{-j})= p_j(t_j, s^w_{-j})$
and $p_j(u_j, s^w_{-j}) = p_j(w_j, s_{-j})$.
Hence for all $s_{-j}$ 
$p_j(v_j, s_{-j}) \leq p_j(w_j, s_{-j})$
and for some $s_{-j}$ such that $s_{-j} \in J_{-j}$ and $s^w_{-j} \in H_{-j}$ (i.e., for some  $s_{-j} \in (J \cap \langle H \rangle)_{-j}$)
$p_j(v_j, s_{-j}) < p_j(w_j, s_{-j})$.
This establishes the claim.

\medskip

\NI \emph{Case 2.} $j \neq \turn(v)$.

Let $i = \turn(v)$. Take some $s_{-j}$.  If $s_i(v) = w$, then
$p_j(v_j, s_{-j})= p_j(t_j, s^w_{-j})$ and $p_j(w_j, s_{-j})= p_j(u_j,
s^w_{-j})$. Thus $p_j(v_j, s_{-j}) \leq p_j(w_j, s_{-j})$ by the choice
of $u_j$ and $w_j$.  Further, if $s_i(v) \neq w$, then $p_j(v_j,
s_{-j})= p_j(w_j, s_{-j})$ by the choice of $w_j$.  So for all
$s_{-j}$ we have $p_j(v_j, s_{-j}) \leq p_j(w_j, s_{-j})$.

Choose an arbitrary $s_{-j}$ such that $s_i(v) = w$, $s^w_{-j} \in
H_{-j}$, and $p_j(t_j, s^w_{-j}) < p_j(u_j, s^w_{-j})$. Additionally,
we can claim that $s_{-j} \in J_{-j}$ by stipulating that $s_{-j} =
t_{-j}$ for an arbitrary joint strategy $t$ from $J$.  Then $s_{-j}
\in (J \cap \langle H \rangle)_{-j}$.

By the choice of $s_i$ we have $p_j(v_j, s_{-j})= p_j(t_j, s^w_{-j})$
and $p_j(w_j, s_{-j})= p_j(u_j, s^w_{-j})$, so $p_j(v_j, s_{-j}) <
p_j(w_j, s_{-j})$.  This establishes the claim for this case.
\end{proof}

We continue with some lemmas concerned with the relation $\to^{\rho}$.

\begin{lemma} \label{lem:subgame2}
Consider a direct subgame $G^w$ of $G$.  Suppose that for some
sequence $\rho$ of sets of strategies of players in $G^w$ and a
subgame $H$ of $\Gamma(G^{w})$, $\Gamma(G^{w}) \to^{\rho} H$.
Suppose further that the subgame $J$ of $\Gamma(G)$ does not depend on $G^w$.
Then $ J \to^{[\rho]} J \cap \langle H \rangle$.  
\end{lemma}

\begin{proof}
We proceed by transfinite induction on the length $\gamma$ of $\rho =
(\rho_{\alpha}, \alpha < \gamma)$.

\medskip

\NI
\emph{Case 1.} $\gamma = 1$.

By Lemma \ref{lem:lifting3}
$J \cap \langle \Gamma(G^w) \rangle \to^{[\rho_0]} J \cap \langle H \rangle$,
so the claim holds since $\langle \Gamma(G^w) \rangle  = \Gamma(G)$ and $J \cap \Gamma(G) = J$.

\medskip

\NI
\emph{Case 2.} $\gamma$ is a successor ordinal $>1$.

Suppose $\gamma = \delta + 1$. Then $\rho = (\rho', \rho_{\delta})$,
where $\rho' := (\rho_{\alpha}, \alpha < \delta)$.  By definition for
some $H'$ we have $ \Gamma(G^{w}) \to^{\rho'} H'$ and
$H' \to^{\rho_{\delta}} H$.
By the induction hypothesis $J  \to^{[\rho']} J \cap \langle H' \rangle$ and by 
Lemma \ref{lem:lifting3} $J \cap \langle H' \rangle \to^{[\rho_{\delta}]} J \cap \langle H \rangle$, so the
claim follows by Note \ref{not:one}$(i)$, since
$[\rho] = ([\rho'], [\rho_{\delta}])$.

\medskip

\NI
\emph{Case 3.} $\gamma$ is a limit ordinal.

By definition for some games $H^{\beta}$, where $\beta < \gamma$, we have
$\Gamma(G^w) \to^{\rho^\beta} H^{\beta}$ and $H =  \bigcap_{\beta < \gamma} H^{\beta}$,
where---recall---$\rho^{\beta} = (\rho_{\alpha}, \alpha < \beta)$.
By the induction hypothesis for all $\beta < \gamma$, we have
$ J \to^{[\rho^{\beta}]} J \cap \langle H^{\beta} \rangle$. 
So by definition
$J \to^{[\rho]} J \cap \langle H \rangle$,
since
$J \cap \langle H \rangle =    \bigcap_{\beta < \gamma} \langle J \cap H^{\beta} \rangle$ as
$\langle H \rangle =    \bigcap_{\beta < \gamma} \langle H^{\beta} \rangle$. 
\end{proof}

\begin{restatable}{lemma}{lemTwo}
  \label{lem:two}
Consider an extensive game $G$ with the root $v$. Suppose that
$(w_{\alpha}, \alpha < \gamma)$ is a sequence of children of $v$ and
that for all $\alpha < \gamma$, $\rho_{\alpha}$ is a sequence of sets
of strategies in the direct subgame $G^{w_{\alpha}}$. Suppose further
that for each $\alpha < \gamma$
$\Gamma(G^{w_{\alpha}}) \to^{\rho_{\alpha}} H^{w_{\alpha}}$,
  where each game $H^{w_{\alpha}}$ is non-empty.
  Let $\rho$ be the concatenation of the sequences
  $(\rho_{\alpha}, \alpha < \gamma)$.  Then
  $\Gamma(G) \to^{[\rho]} \bigcap_{\alpha < \gamma} \langle H^{w_{\alpha}} \rangle$.
\end{restatable}

By assumption each $H^{w_{\alpha}}$ is a non-empty subgame of
$\Gamma(G^{w_{\alpha}})$, so each $\langle H^{w_{\alpha}} \rangle$ is
a non-empty subgame of $\Gamma(G)$, and consequently
$\bigcap_{\alpha < \gamma} \langle H^{w_{\alpha}} \rangle$ is also a non-empty subgame
of $\Gamma(G)$.

Informally, suppose that for each direct subgame $G^{w_{\alpha}}$ of
$G$ we can reduce the corresponding strategic game
$\Gamma(G^{w_{\alpha}})$ to a non-empty game $H^{w_{\alpha}}$.  Then
the strategic game $\Gamma(G)$ can be reduced to a strategic game the
strategies of which are obtained by intersecting for each player the
extensions of his strategy sets in all games $H^{w_{\alpha}}$. To
establish this lemma we do not assume that
$(w_{\alpha}, \alpha < \gamma)$ contains all children of $v$, which
makes it possible to proceed by induction.

\begin{proof}
We proceed by transfinite induction on the length $\gamma$ of $\rho$.

\medskip

\NI
\emph{Case 1.} $\gamma = 1$. Follows from Lemma \ref{lem:subgame2} with $J = \Gamma(G)$.

\medskip

\NI
\emph{Case 2.} $\gamma$ is a successor ordinal $>1$.

Suppose $\gamma = \delta + 1$. By the induction hypothesis
$\Gamma(G) \to^{[\rho^{\delta}]} \bigcap_{\alpha < \delta} 
\langle H^{w_{\alpha}} \rangle$, 
where $\rho^{\delta}$ is the concatenation of the sequences
$(\rho_{\alpha}, \alpha < \delta)$.
We also have by assumption 
$\Gamma(G^{w_{\delta}}) \to^{\rho_{\delta}} H^{w_{\delta}}$.

Note that the subgame $\bigcap_{\alpha < \delta} \langle H^{w_{\alpha}} \rangle$ of $\Gamma(G)$ does not depend on $G^{w_{\delta}}$,
so by Lemma \ref{lem:subgame2} we have that 
$\bigcap_{\alpha < \delta} \langle H^{w_{\alpha}} \rangle \to^{[\rho_{\delta}]} \bigcap_{\alpha < \delta} \langle H^{w_{\alpha}} \rangle \cap \langle H^{w_{\delta}} \rangle$.
By Note \ref{not:one}$(i)$ the claim follows.

\medskip

\NI
\emph{Case 3.} $\gamma$ is a limit ordinal.

By the induction hypothesis for all $\beta < \gamma$ 
$\Gamma(G) \to^{[\rho^\beta]} \bigcap_{\alpha < \beta} \langle H^{w_{\alpha}} \rangle$,
where $\rho^{\beta}$ is the concatenation of the sequences
$(\rho_{\alpha}, \alpha < \beta)$.
Then by Note \ref{not:one}$(ii)$ and by definition
$\Gamma(G) \to^{[\rho]} \bigcap_{\beta < \gamma} \bigcap_{\alpha < \beta} \langle H^{w_{\alpha}} \rangle$.
But
$\bigcap_{\beta < \gamma} \bigcap_{\alpha < \beta} \langle  H^{w_{\alpha}} \rangle
  = \bigcap_{\alpha < \gamma} \langle  H^{w_{\alpha}} \rangle$,
so the claim follows.
\end{proof}

The next lemma shows that when each subgame $H^{w_{\alpha}}$ of
$\Gamma(G^{w_{\alpha}})$ is trivial, under some natural assumptions
the subgame $\bigcap_{\alpha < \gamma} \langle H^{w_{\alpha}} \rangle$
of $\Gamma(G)$ can then be reduced in one step to a trivial game.

\begin{restatable}{lemma}{lemTrivial}
    \label{lem:trivial}
    Consider an extensive game $G$ with the root $v$.  Suppose that
  \begin{enumerate}[(a)]
\item $G$ has a subgame perfect equilibrium and  all subgame perfect equilibria of $G$ are payoff equivalent,
\item for all $w \in C(v)$, $\spe(G^{w}) \sse H^{w}$, where $H^{w}$ is
  a trivial subgame of $\Gamma(G^{w})$. 
  \end{enumerate}
\noindent Then for some set of strategies $A$ we have
  $\bigcap_{w \in C(v)} \langle H^{w} \rangle \to^A H'$, where $H'$ a
  trivial game and $\spe(G) \sse H'$. 
\end{restatable}

\begin{proof}
  Let $H := \bigcap_{w \in C(v)} \langle H^{w} \rangle$. Note that $H$
  is a non-empty subgame of $\Gamma(G)$.  

Denote the unique outcome in the game $H^{w}$ by $val^{w}$, i.e., for
all joint strategies $s$ in $H^{w}$ we have $p(s) = val^{w}$. Then the
possible outcomes in $H$ are $val^{w}$, where $w \in C(v)$.  More
precisely, suppose that $i = \turn(v)$. Then if $s$ is a joint
strategy in $H$, then $p(s) = val^{w}$, where $s_i(v) = w$.

Take two strategies $t'_i$ and $t''_i$ of player $i$ in $H$ with
$t'_i(v) = w_1$ and $t''_i(v) = w_2$ such that
$val_i^{w_1} < val_i^{w_2}$.  This means that for any joint strategies
$s_{-i}$ from $H_{-i}$ we have $p_i(t'_i, s_{-i} < p_i(t''_i, s_{-i}$,
so $t'_i$ is weakly dominated in $H$ by $t''_i$ (actually, even
strictly dominated).

By assumption \emph{(a)} $G$ has a subgame perfect equilibrium, so
by Corollary 7 of \cite{AS21}
$\max \{val_i^{w} \mid w \in C(v)\}$ exists.  Denote it by $val_i$ and
let $W := \{w \in C(v) \mid val_i^{w} = val_i \}$. So $W$ is the set
of children $w$ of $v$ for which the corresponding value $val_{i}^{w}$
is maximal.  Finally, let $A$ be the set of strategies $t_i$ of player
$i$ in $H$ such that $t_i(v) \not\in W$.

By the above observation about $t'_i$ and $t''_i$ all strategies in
$A$ are weakly dominated in $H$.  By removing them from $H$ we get a
game $H'$ with the unique payoff $val_i$ for player $i$.  To prove
that $H'$ is trivial consider two joint strategies $s$ and $t$ in
$H'$.  Suppose that $s_i(v) = w_1$ and $t_i(v) = w_2$.  Then
$w_1, w_2 \in W$, $s^{w_1} \in H^{w_1}$, $t^{w_2} \in H^{w_2}$,
$p(s) = p(s^{w_1})$, and $p(t) = p(t^{w_2})$.
  
By Theorem 8 of \cite{AS21} subgame perfect equilibria $u'$ and $u''$
in $G$ exist such that $u'_i(v) = w_1$, $(u')^{w_1}$ is a subgame
perfect equilibrium in $G^{w_1}$, $u''_i(v) = w_2$, and $(u'')^{w_2}$
is a subgame perfect equilibrium in $G^{w_2}$.  Then
$p(u') = p((u')^{w_1})$ and $p(u'') = p((u'')^{w_2})$, so
$p((u')^{w_1}) = p((u'')^{w_2})$ by assumption \emph{(a)}.
Further, by assumption \emph{(b)} both $(u')^{w_1} \in H^{w_1}$ and
$(u'')^{w_2} \in H^{w_2}$, so since both subgames are trivial,
$p(s^{w_1}) = p((u')^{w_1})$ and $p(t^{w_2}) = p((u')^{w_2})$.
Consequently $p(s) = p(t)$, which proves that $H'$ is trivial.

To prove that $\spe(G) \sse H'$
consider a subgame perfect equilibrium $s$
in $G$.  Take some $u \in C(v)$. By assumption \emph{(b)},
$s^u \in H^u$, so $p_i(s^u) = val_i^u$ and, by the definition of
$\langle \cdot \rangle$, $s \in H$.  Suppose that $s_i(v)=w$.  By
Corollary 7 of \cite{AS21} $val_i^w = val_i$, i.e., $s_i(v) \in W$.
This means that $s_i \not\in A$ and thus $s \in H'$.
\end{proof}

\section{SPE-invariant games}
\label{sec:main}
We can now  prove the desired result. 

\begin{theorem}
  \label{thm:trivial}
  Consider an SPE-invariant extensive game $G$.
There exists a sequence $\rho$ of strategies of
players in $G$ and a subgame $H$ of $\Gamma(G)$ such that $\Gamma(G) \to^\rho H$, $H$ is trivial and $\spe(G) \subseteq H$.
\end{theorem}

\begin{proof}
  We proceed by induction on the rank of the game tree of $G$.
  For game trees of rank 0 all strategies are empty functions, so
  $\Gamma(G)$ is a trivial game with the unique joint strategy
  $(\ES, \LL, \ES)$ and $\spe(G) = \{(\ES, \LL, \ES)\}$, so the claim
  holds.  Suppose that the rank of the game tree of $G$ is
  $\alpha > 0$ and assume that claim holds for all extensive games
  with the game trees of rank smaller than $\alpha$.

  Let $v$ be the root of $G$. Each direct subgame of $G$ is
  SPE-invariant, so by the induction hypothesis for all $w \in C(v)$
  there exists a sequence $\rho^w$ of strategies of players in $G^w$
  and a subgame $H^w$ of $\Gamma(G^w)$ such that $\Gamma(G^w)
  \to^{\rho^w} H^w$, $H^w$ is trivial and $\spe(G^w) \sse H^w$.
The claim now follows by Lemmas \ref{lem:two} and \ref{lem:trivial}.
\end{proof}

\medskip

The following example illustrates the use of this theorem.
An extensive game is called \bfe{generic} if each payoff function is an
injective. 

\begin{figure}[ht]
\centering
\begin{minipage}{.4\textwidth}
  \centering  
  \tikzstyle{level 1}=[level distance=1.2cm, sibling distance=3cm]
  \tikzstyle{level 2}=[level distance=1.2cm, sibling distance=2.5cm]
\begin{tikzpicture}
 \node (r){$v(1)$}
 child{
   node (a){$S_1:(1,0)$}
   edge from parent
   node[left]{\scriptsize $S$}
   }
 child{
   node (b){$C_1(2)$}
   child{
     node (f){\small $S_2:(0,2)$}
     edge from parent
     node[left]{\scriptsize $S$}
   }
   child{
     node (g){\small $C_2:(2,1)$}
     edge from parent
     node[right]{\scriptsize $C$}
     edge from parent
   }
   edge from parent
   node[right]{\scriptsize $C$}
 };

\end{tikzpicture}
  
\captionof{figure}{Centipede game with 2 periods}
  \label{fig:centipede}
\end{minipage}%
\begin{minipage}{.6\textwidth}
  \centering
    \tikzstyle{level 1}=[level distance=1.2cm, sibling distance=3.5cm]
  \tikzstyle{level 2}=[level distance=1.2cm, sibling distance=3.5cm]
\begin{tikzpicture}
 \node (r){$C_{2t}(1)$}
 child{
   node (a){\small $S_{2t+1}:(x+1,y)$}
   edge from parent
   node[left]{\scriptsize $S$}
   }
 child{
   node (b){$C_{2t+1}(2)$}
   child{
     node (f){\small $S_{2t+2}:(x,y+3)$}
     edge from parent
     node[left]{\scriptsize $S$}
   }
   child{
     node (g){\small $C_{2t+2}:(x+2,y+2)$}
     edge from parent
     node[right]{\scriptsize $C$}
     edge from parent
   }
   edge from parent
   node[right]{\scriptsize $C$}
 };

\end{tikzpicture}

\captionof{figure}{From $t$ to $t+2$ periods}
 \label{fig:test2}
  
\end{minipage}
\end{figure}

\begin{example}
\label{exa:centipede}  
Recall that the centipede game, introduced in \cite{Ros81} (see also
\cite[pages 106-108]{OR94}), is a two-players extensive game played for
an even number of periods.  We define it inductively as follows.  The
game with 2 periods is depicted in Figure \ref{fig:centipede}. Here and
below the argument of each non-leaf is the player whose turn is to
move, and the leaves are followed by players' payoffs. The moves are
denoted by the letters $C$ and $S$.
The game with $2t+2$ periods is obtained from the game with $2t$
periods by replacing the leaf $C_{2t}$ by the tree depicted in Figure
\ref{fig:test2}.

By the the result of \cite[pages
108-109]{OR94}) each centipede game can be reduced by an iterated
elimination of weakly dominated strategies to a trivial game which
contains the unique subgame perfect equilibrium, with the outcome
$(1,0)$.
We now show that the same holds for an infinite version of the
centipede game $G$ in which player 2 begins the game by selecting an
even number $2t > 0$.  Subsequently, the centipede version with $2t$
periods is played.

Note that $G$ is SPE-invariant. Indeed, $G$ has infinitely many
subgame perfect equilibria (one for each first move of player 2), but
each of them yields the outcome $(1,0)$. Moreover, each subgame of $G$
is either a centipede game with $2t$ periods for some $t > 0$, or a
subgame of such a game. So each subgame of $G$ is a finite generic
game and thus has a unique subgame perfect equilibrium.

By Theorem \ref{thm:trivial} we can reduce $G$ by an infinite iterated
elimination of weakly dominated strategies to a trivial game which
contains all its subgame perfect equilibria.
Note that the strategy elimination sequence constructed in the proof of this theorem 
consists of for more than $\omega$ steps.
\HB
\end{example}

For finite extensive games, Theorem \ref{thm:trivial} extends the
original result reported in \cite[pages 108-109]{OR94}.  Namely, the
authors prove the corresponding result for finite extensive games that
are generic. In such games a unique subgame perfect equilibrium
exists, while we only claim that the game is SPE-invariant.

To clarify the relevance of this relaxation let us mention two classes
of well-founded extensive games that are SPE-invariant and that were
studied for finite extensive games. Following \cite{Bat97} we say that
an extensive game $(T, \turn, p_1, \LL, p_n)$ is \bfe{without relevant
  ties} if for all non-leaf nodes $u$ in $T$ the payoff function
$p_i$, where $\turn(u)=i$, is injective on the leaves of $T^u$.
This
is a more general property than being generic. The relevant property
for finite extensive games is that a game without relevant ties has a
unique subgame perfect equilibrium, see \cite{AS21a} for a
straightforward proof. In the case of well-founded games a direct
modification of this proof, that we omit, shows that every extensive
game without relevant ties has at most one subgame perfect
equilibrium. Further, if a game is without relevant ties, then so is
every subgame of it, so we conclude that well-founded games without
relevant ties are SPE-invariant.

Next, following \cite{MS97} we say that an extensive game
$(T, turn, p_1, \LL, p_n)$ satisfies the \bfe{transference of
  decisionmaker indifference (TDI)} condition if:
\[
  \begin{array}{l}
\mbox{$\fa i \in \{1, \ldots, n\} \: \fa r_i, t_i \in S_i \: \fa s_{-i} \in S_{-i}$} \\
\mbox{[$p_{i}(\leaf(r_i, s_{-i})) = p_{i}(\leaf(t_i, s_{-i})) \to
p(\leaf(r_i, s_{-i})) = p(\leaf(t_i, s_{-i}))$].}
  \end{array}
\]
where $S_i$ is the set of strategies of player $i$. 
Informally, this condition states that whenever for some player $i$, two of his strategies
$r_i$ and $t_i$ are indifferent w.r.t.~some joint strategy $s_{-i}$ of the other players
then this indifference extends to all players.

Strategic games that satisfy the TDI condition are of interest because of the main result of \cite{MS97} which states
that in finite games that satisfy this condition iterated elimination of weakly dominated strategies is order 
independent.\footnote{Alternative proofs of this result were given in \cite{Apt04} and \cite{Ost05}.}
The authors also give examples of natural games that satisfy this condition. Also strictly competitive games studied in the
next section satisfy this condition.

The following result extends an implicit result of \cite{MS97} to well-founded games.

\begin{theorem}
  \label{thm:TDIinvariance}
  Consider an extensive game $G$. Suppose that
  $G$ has finitely many outcomes and $G$ satisfies the TDI condition. Then $G$ is SPE-invariant.
\end{theorem}
\begin{proof}
   We reduce the game $G$ to a finite game $H$ as follows. First,
   consider the set of all leaves of the game tree $T$ of $G$ that are
   the ends of the plays corresponding with a subgame perfect
   equilibrium. Next, for each outcome associated with a subgame
   perfect equilibrium retain in this set just one leaf with this
   outcome. By assumption the resulting set $L$ is finite.

   Next, order the leaves arbitrarily. Following this ordering
   remove all leaves with an outcome already associated with an
   earlier leaf, but ensuring that the leaves from $L$ are retained.  Let
   $M$ be the resulting set of leaves.  Finally, remove all nodes of
   $T$ from which no leaf in $M$ can be reached.

   The resulting tree corresponds to a finite extensive game $H$ in
   which all the outcomes possible in $G$ are present.  Further, all
   the leaves of $H$ are also leaves of $G$, so $H$ satisfies the TDI
   condition since $G$ does.  So by Theorem 12 of \cite{AS21a} (that is
   implicit in \cite{MS97}) all subgame perfect equilibria of $H$ are
   payoff equivalent.

   Further, by Theorem \ref{thm:constantSPE} $G$ has a subgame perfect
   equilibrium.  Consider two subgame perfect equilibria $s$ and
   $t$ in $G$ with the outcomes $p(s)$ and $p(t)$. By construction two
   subgame perfect equilibria $s'$ and $t'$ in $H$ exist such that
   $p(s) = p(s')$ and $p(t) = p(t')$. We conclude that all subgame
   perfect equilibria of $G$ are payoff equivalent.

   To complete the proof it suffice to note that if an extensive game
   $G$ satisfies the TDI condition, then so does every subgame of it.
   Indeed, consider a subgame $G^w$ of $G$. Let $i = \turn(w)$ and
   take $r^w_i, t^w_i \in S^w_i$ and $s^w_{-i} \in S^w_{-i}$. Extend
   these strategies to the strategies $r_i, t_i$ and $s_{-i}$ in the
   game $G$ in such a way that $w$ lies both on $\play(r_i, s_{-i})$
   and on $\play(t_i, s_{-i})$. Then
   $p(r^w_i, s^w_{-i}) = p(r_i, s_{-i})$ and
   $p(t^w_i, s^w_{-i}) = p(t_i, s_{-i})$, so the claim follows.
 \end{proof}
 
 \begin{corollary}
The claim of Theorem \ref{thm:trivial} holds for extensive games
with  finitely many outcomes that satisfy the TDI condition.
\end{corollary}

\NI
\textbf{Conjecture} Every extensive game that satisfies the TDI-condition
   is SPE-invariant. 
\II

If the conjecture is true, Theorem \ref{thm:trivial} holds for all
extensive games that satisfy the TDI condition.  An example of a game
with infinitely many outcomes that satisfies the TDI condition is the
infinite version of the centipede game from Example
\ref{exa:centipede}.


\section{Strictly competitive extensive games}
\label{sec:strictly}
In some games, for instance, the infinite version of the centipede
game from Example \ref{exa:centipede}, infinite rounds of elimination
of weakly dominated strategies are needed to solve the game.  In this
section, we focus on maximal elimination of weakly dominated
strategies and identify a subclass of extensive games for which we can
provide a finite bound on the number of elimination steps required to
solve the game. The outcome is our second main result which is a generalization 
of the following result due to \cite{Ewe02} to a class of well-founded games.
\II

\NI
\textbf{Theorem} Every finite extensive zero-sum game with $n$ outcomes can be reduced to
a trivial game by the maximal iterated elimination of weakly
dominated strategies in $n-1$ steps.
\II

We first present some auxiliary results. Their proofs follow our
detailed exposition in \cite{AS21a} of the proofs in \cite{Ewe02}
generalized to strictly competitive games, now appropriately modified
to infinite games.

\subsection{Preliminary results}
We denote by $H^1$ the subgame of $H$ obtained by the elimination of
all strategies that are weakly dominated in $H$, and put $H^{0} := H$
and $H^{k+1} := (H^k)^1$, where $k \geq 1$.  Abbreviate the phrase
`iterated elimination of weakly dominated strategies' to IEWDS.  If
for some $k$, $H^k$ is a trivial game we say that $H$ \bfe{can be
  solved by the IEWDS}. 


In infinite strategic games with finitely many outcomes it is possible
that all strategies of a player are weakly dominated as shown in the
Example \ref{exa:0-1}.
Then by definition, $H^1$ is an empty game. We define a class of
games, called WD-admissible games in which this does not
happen.

\begin{example} \label{exa:0-1}  
Consider the following infinite zero-sum strategic game with 
  two outcomes:
  
\begin{center}
\begin{game}{5}{5}
        & $A$    & $B$      & $C$   & $D$    &$\dots$ \\
$A$     &$0,0$   &$0,0$     &$0,0$  &$0,0$   &$\dots$ \\
$B$     &$0,0$   &$1,-1$    &$0,0$  &$0,0$   &$\dots$ \\
$C$     &$0,0$   &$1,-1$    &$1,-1$ &$0,0$   &$\dots$ \\
$D$     &$0,0$   &$1,-1$   &$1,-1$  &$1,-1$   &$\dots$ \\
$\dots$ & $\dots$ & $\dots$ & $\dots$ & $\dots$ & $\dots$
\end{game}
\end{center}  

This game has a Nash equilibrium, namely $(A,A)$, but each strategy of
the row player is weakly dominated. So after one round of elimination
the empty game is reached.
\HB
\end{example}

Consider a strategic game $H$. We say that a strategy is
\bfe{undominated} if no strategy weakly dominates it. Next, we say
that $H$ is \bfe{WD-admissible} if for all subgames $H'$ of it the
following holds: \textit{each strategy is undominated or is weakly
  dominated by an undominated strategy}.
%
Intuitively, a strategic game $H$ is WD-admissible if in every subgame
$H'$ of it, for every strategy $s_i$ in $H'$ the relation `is weakly
dominated' in $H'$ has a maximal element above $s_i$. The crucial
property of WD-admissible games is formalised in the following lemma
whose proof follows directly by induction.

  \begin{lemma}
  \label{lem:weakPI}
  Let $H := (H_1, \LL, H_n, p_1, \LL, p_n)$ be a WD-admissible
  strategic game and for $k \geq 1$, let $H^k := (H^k_1, \LL, H^k_n, p_1, \LL, p_n)$. Then
$\fa i \in \{1, \LL, n\} \ \fa s_i \in H_i \ \te t_i \in H^k_i \ \fa s_{-i} \in H^k_{-i}:
p_i(t_i,s_{-i}) \geq  p_i(s_i,s_{-i})$.
\end{lemma}

A two player strategic game $H=(H_1, H_2, p_1,
p_2)$ is called \bfe{strictly competitive} if
$
\fa i \in \{1,2\} \ \fa s,s' \in S :
p_i(s) \geq p_i(s') \text{ iff } p_{\ibar}(s) \leq p_{\ibar}(s')$.
For $i \in \{1,2\}$ we define
$maxmin_{i}(H) := \max_{s_{i} \in H_i} \min_{s_{-i} \in H_{-i}} p_i(s_i, s_{-i})$.
We allow $- \infty$ and $\infty$ as minima and maxima, so
$maxmin_{i}(H)$ always exists.  When $maxmin_{i}(H)$ is finite we call
any strategy $s^{*}_i$ such that $\min_{s_{-i} \in H_{-i}} p_i(s^*_i,
s_{-i}) = maxmin_{i}(H)$ a \bfe{security strategy} for player $i$ in
$H$.

We shall reuse the following auxiliary results from \cite{AS21a}.

\begin{note}
  \label{lem:scSamePayoff}
  Let $H = (H_1, H_2, p_1, p_2)$ be a strictly competitive strategic game. Then
\[
\fa i \in \{1,2\} \ \fa s,s' \in S : p_i(s) = p_i(s') \text{ iff } p_{\ibar}(s) = p_{\ibar}(s').
\]
\end{note}

This simply means that every strictly competitive strategic game
satisfies the TDI condition.

\begin{lemma} \label{lem:sc-sum1}
 Consider a strictly competitive strategic game $H$ with a Nash
  equilibrium $s$.  Suppose that for some $i \in \{1,2\}$, $t_i$
  weakly dominates $s_i$.  Then $(t_i, s_{-i})$ is also a Nash
  equilibrium.
\end{lemma}

\begin{lemma} \label{lem:scTwo}
  Consider a strictly competitive strategic game $H$ with two outcomes that has a Nash
  equilibrium. Then $H^1$ is a trivial game.
\end{lemma}

\noindent The following result is standard (for the used formulation see, e.g.,
\cite[Theorem 5.11, page 235]{Rit02}).

\begin{theorem}
  \label{thm:scminimax}
  Consider a strictly competitive strategic game $H$.
\begin{enumerate}[(i)]
\item All Nash equilibria of $H$ yield the same payoff for player $i$,
  namely $maxmin_{i}(H)$.
\item All Nash equilibria of $H$ are of the form  $(s^{*}_1, s^{*}_2)$
where each $s^{*}_i$ is a security strategy for player $i$.  
\end{enumerate}
\end{theorem}

By modifying the proof of Corollary 5 from \cite{AS21a}
appropriately, we have the following.

\begin{restatable}{lemma}{lemHOne}
  \label{lem:h1} Consider a WD-admissible strictly
  competitive strategic game $H$ that has a Nash equilibrium. Then $H^1$
  has a Nash equilibrium, as well, and for all $i \in \{1,2\}$,
  $maxmin_{i}(H) = maxmin_{i}(H^1)$.
\end{restatable}

\subsection{A bound on IEWDS}
\label{subsec:scbound}

We now move on to a discussion of extensive games. 
We say that an
extensive game $G$ is \bfe{WD-admissible} (respectively, \bfe{strictly competitive}) if $\Gamma(G)$ is WD-admissible (respectively, strictly competitive).
We write $\Gamma^{k}(G)$ instead of
$(\Gamma(G))^k$, $\Gamma_i(G)$ instead of $(\Gamma(G))_i$, and
$\Gamma^k_i(G)$ instead of $(\Gamma^k(G))_i$.  So
$\Gamma^{0}(G) = \Gamma(G)$.
Further, for a strictly competitive game $H = (H_1, H_2, p_1, p_2)$
with finitely many outcomes for each player $i$ we define the
following three sets:
$p_i^{\pmax}(H) := \max_{s \in S} p_i(s)$,
$\wmax_i(H) : =\{s_i \in H_i \mid \fa s_{-i} \in H_{-i} \:
p_i(s_i,s_{-i})=p_i^{\pmax}(H)\}$ and $\lose_{-i}(H) = \{s_{-i} \in
H_{-i} \mid \exists s_{i} \in H_{i} \: p_{i}(s_{i},s_{-i}) =
p^{\pmax}_{i}(H)\}$. By the assumption about $H$, $p_i^{\pmax}(H)$ is
finite.


We can then prove the following generalization of the crucial Lemma 1 and Theorem 1 from \cite{Ewe02}, where the proofs are analogous to that of Lemma 18 and Theorem 19 in \cite{AS21a}.

\begin{restatable}{lemma}{lemscLose}
  \label{lem:scLose}
  Let $G$ be a WD-admissible strictly competitive extensive
  game with finitely many outcomes.  For all $i \in \{1,2\}$ and for
  all $k \geq 0$, if $\wmax_i(\Gamma^k(G)) = \emptyset$ then 
    $\lose_{\ibar}(\Gamma^k(G)) \cap \Gamma^{k+2}_{-i}(G) = \ES$.
\end{restatable}

  Lemma~\ref{lem:scLose} implies that if for all $i \in \{1,2\}$,
  $\wmax_i(\Gamma^k(G)) = \emptyset$ then two further rounds of
  eliminations of weakly dominated strategies remove from
  $\Gamma^k(G)$ at least two outcomes. 

This allows us to establish the following result. The proof is almost the same as the one given in \cite[Theorem 19]{AS21a} for the finite extensive games. We reproduce it here
for the convenience of the reader.

\begin{restatable}{theorem}{thmSciewds}
\label{thm:sciewds}
Let $G$ be a WD-admissible strictly competitive
extensive game with at most $m$ outcomes. Then $\Gamma^{m-1}(G)$ is a trivial
game.
\end{restatable}

\begin{proof}
We prove a stronger claim, namely that for all $m \geq 1$ and $k \geq 0$
if $\Gamma^{k}(G)$ has at most $m$ outcomes, then $\Gamma^{k+m-1}(G)$ is a
trivial game.  

We proceed by induction on $m$. For $m = 1$ the claim is trivial.
For $m = 2$ we first note that by Theorem~\ref{thm:constantSPE} and Lemma \ref{lem:h1} each
game $\Gamma^{k}(G)$ has a Nash equilibrium. So the claim follows by Lemma \ref{lem:scTwo}.
For $m > 2$ two cases arise.
\II

\NI
\emph{Case 1.} For some $i \in \{1,2\}$, $\wmax_i(\Gamma^k(G)) \neq \emptyset$.

For player $i$ every strategy $s_i \in \wmax_i(\Gamma^k(G))$ weakly
dominates all strategies $s_i' \notin \wmax_i(\Gamma^k(G))$ and no 
strategy in $\wmax_i(\Gamma^k(G))$ is weakly dominated. So 
the set of strategies of player $i$ in $\Gamma^{k+1}(G)$ equals
$\wmax_i(\Gamma^k(G))$ and consequently $p_i^{\pmax}(\Gamma^k(G))$ is his unique payoff
in this game. By  Note \ref{lem:scSamePayoff} $\Gamma^{k+1}(G)$, and hence also
$\Gamma^{k+m-1}(G)$, is a trivial game. 
\II

\NI
\emph{Case 2.} For all $i \in \{1,2\}$, $\wmax_i(\Gamma^k(G)) = \emptyset$.

Take joint strategies $s$ and $t$ such that
$p_1(s)=p_1^{\max}(\Gamma^k(G))$ and $p_2(t)=p_2^{\max}(\Gamma^k(G))$.
By Note \ref{lem:scSamePayoff} 
the outcomes $(p_1(s), p_2(s))$ and $(p_1(t), p_2(t))$ are different since $m > 1$.

We have $s_2 \in \lose_{2}(\Gamma^k(G))$ and $t_1 \in \lose_{1}(\Gamma^k(G))$.
Hence by Lemma \ref{lem:scLose} for no joint strategy $s'$ in $\Gamma^{k+2}(G)$
we have $p_1(s')=p_1^{\max}(\Gamma^k(G))$ or  $p_2(s')=p_2^{\max}(\Gamma^k(G))$.

So $\Gamma(G^{k+2})$ has at most $m-2$ outcomes.  By the induction
hypothesis $\Gamma(G^{k+m-1})$ is a trivial game.
\end{proof}

\II

We now show that Theorem \ref{thm:sciewds} holds for a large class of natural games.
Call an extensive game \bfe{almost constant} if for all but finitely
many leaves the outcome is the same.
Note that every almost constant game has finitely many outcomes, but
the converse does not hold. Indeed, it suffices to take a game with
two outcomes, each associated with infinitely many leaves. 
The following general result holds.

\begin{theorem} \label{thm:almost}
Every almost constant extensive game is WD-admissible.
\end{theorem}

\begin{proof}
We begin with two unrelated observations. Call a function $p: A \to B$ \bfe{almost constant} if for some $b$ we
have $p(a) = b$ for all but finitely many $a \in A$.

\medskip
\NI
\emph{Observation 1.} Consider two
sequences of some elements $(v_0, v_1, \LL)$ and $(w_0, w_1, \LL)$
such that $v_j \neq v_k$, $v_j \neq w_k$, and $w_j \neq w_k$ for
all $j \geq 0$ and $k > j$, and a function
$
p: \{v_0, v_1, \LL \} \cup \{w_0, w_1, \LL \} \to B
$
such that $p(v_j) \neq p(w_j)$ for all $j \geq 0$.
Then $p$ is not almost constant.

\noindent Indeed, otherwise for some $k \geq 0$ the function
$
  p: \{v_{k}, v_{k+1}, \LL \} \cup \{w_{k}, w_{k+1}, \LL \} \to B
$
would be constant.

\medskip
\NI
\emph{Observation 2.} Take an extensive game.  For some player $i$,
consider two joint strategies $(s_i, s_{-i})$ and $(s_i',s'_{-i})$.
If $\leaf(s_i,s_{-i}) = \leaf(s'_i,s'_{-i})$ then
$\leaf(s_i,s_{-i}) = \leaf(s'_i,s_{-i})$.

Indeed, consider any node $w$ in $\mathit{play}(s_i,s_{-i})$ such
that $\mathit{turn}(w)=i$. Then by assumption $s_i(w)=s'_i(w)$. This
implies that $\mathit{play}(s_i,s_{-i}) = \mathit{play}(s'_i,s_{-i})$,
which yields the claim.

Now consider an almost constant extensive game $G$.  Take an arbitrary
subgame $H$ of $\Gamma(G)$.  Suppose by contradiction that for some
player $i$ there exists an infinite sequence of strategies $s_i^0,
s_i^1, s_i^2, \ldots$ such that for all $j \geq 0$, $s_i^{j+1}$ weakly
dominates $s_i^{j}$ in $H$. By definition of weak dominance, for all
$j \geq 0$ there exists $s^j_{-i} \in H_{-i}$ such that
$p_i(s_i^j,s_{-i}^j) < p_i(s_i^{j+1},s_{-i}^j)$.
Let for $j \geq 0$, $v_j = \leaf(s_i^j,s_{-i}^j)$ and
$w_j = \leaf(s_i^{j+1},s_{-i}^j)$.  By the above inequalities
$p_i(v_j) \neq p_i(w_j)$ for all $j \geq 0$.  

We now argue that $v_j \neq v_k$, $v_j \neq w_k$, and $w_j \neq w_k$
for all $j \geq 0$ and $k > j$.  First, note that by the transitivity of
the `weakly dominates' relation we have the following.

\begin{itemize}
\item $p_i(s_i^j, s_{-i}^j) <  p_i(s_i^{j+1}, s_{-i}^j) \leq p_i(s_i^{k},s_{-i}^j)$,
\item  $p_i(s_i^j, s_{-i}^j) < p_i(s_i^{j+1}, s_{-i}^j) \leq  p_i(s_i^{k+1},s_{-i}^j)$,
\item  $p_i(s_i^{j+1}, s_{-i}^k) \leq p_i(s_i^{k},s_{-i}^k) < p_i(s_i^{k+1},s_{-i}^k)$.
\end{itemize}
%
This implies in turn,
$\leaf(s_i^{j},s_{-i}^j) \neq \leaf(s_i^{k},s_{-i}^j)$,
$\leaf(s_i^{j},s_{-i}^j) \neq \leaf(s_i^{k+1},s_{-i}^j)$, and
$\leaf(s_i^{j+1}, s_{-i}^k) \neq \leaf(s_i^{k+1},s_{-i}^k)$.
So by Observation 2
we have the following.
\begin{itemize}
\item $v_j = \leaf(s_i^{j},s_{-i}^j) \neq \leaf(s_i^{k},s_{-i}^k) = v_k$,
\item $v_j = \leaf(s_i^{j},s_{-i}^j) \neq \leaf(s_i^{k+1},s_{-i}^k) = w_k$,
\item $w_j = \leaf(s_i^{j+1},s_{-i}^j) \neq \leaf(s_i^{k+1},s_{-i}^k) = w_k$.
\end{itemize}

By Observation 1, $p_i$ is not almost constant,
which contradicts the assumption that $G$ is almost constant.  By the transitivity of the `weakly
dominates' relation we conclude that $G$ is WD-admissible.
\end{proof}

\begin{corollary}
\label{cor:almostsc}
Let $G$ be an almost constant strictly competitive
extensive game with at most $m$ outcomes. Then $\Gamma^{m-1}(G)$ is a trivial
game.
\end{corollary}

\subsection*{Acknowledgments}
We thank the reviewers for their helpful comments. The second author
was partially supported by the grant CRG/2022/006140.

 \bibliographystyle{eptcs}

\bibliography{ref-s,e}

\end{document}